\documentclass[12pt]{article}

\usepackage{latexsym,amsmath,amssymb,amsthm,amsfonts,graphicx,natbib}
\usepackage{epsfig}
\usepackage{setspace}
\usepackage{bm}
\newcommand*\patchAmsMathEnvironmentForLineno[1]{%
      \expandafter\let\csname old#1\expandafter\endcsname\csname #1\endcsname
      \expandafter\let\csname oldend#1\expandafter\endcsname\csname end#1\endcsname
      \renewenvironment{#1}%
         {\linenomath\csname old#1\endcsname}%
         {\csname oldend#1\endcsname\endlinenomath}}%
    \newcommand*\patchBothAmsMathEnvironmentsForLineno[1]{%
      \patchAmsMathEnvironmentForLineno{#1}%
      \patchAmsMathEnvironmentForLineno{#1*}}%
    \AtBeginDocument{%
    \patchBothAmsMathEnvironmentsForLineno{equation}%
    \patchBothAmsMathEnvironmentsForLineno{align}%
    \patchBothAmsMathEnvironmentsForLineno{flalign}%
    \patchBothAmsMathEnvironmentsForLineno{alignat}%
    \patchBothAmsMathEnvironmentsForLineno{gather}%
    \patchBothAmsMathEnvironmentsForLineno{multline}%
    }
\usepackage[mathlines,displaymath]{lineno}	

\usepackage[textwidth=6in,textheight=8.5in,centering]{geometry}

\include{def}
\def\dispmuskip{\thinmuskip= 3mu plus 0mu minus 2mu \medmuskip=  4mu plus 2mu minus 2mu \thickmuskip=5mu plus 5mu minus 2mu}
\def\textmuskip{\thinmuskip= 0mu                    \medmuskip=  1mu plus 1mu minus 1mu \thickmuskip=2mu plus 3mu minus 1mu}
\def\beq{\dispmuskip\begin{equation}}    \def\eeq{\end{equation}\textmuskip}
\def\beqn{\dispmuskip\begin{displaymath}}\def\eeqn{\end{displaymath}\textmuskip}
\def\bea{\dispmuskip\begin{eqnarray}}    \def\eea{\end{eqnarray}\textmuskip}
\def\bean{\dispmuskip\begin{eqnarray*}}  \def\eean{\end{eqnarray*}\textmuskip}

\def\paradot#1{\vspace{1.3ex plus 0.7ex minus 0.5ex}\noindent{\bf\boldmath{#1.}}}

\newtheorem{theorem}{Theorem}

\newtheorem{algorithm}{Algorithm}

\newtheorem{assumption}{Assumption}

\newcommand{\KL}{Kullback-Leibler}

\newcommand{\wh}{\widehat}
\newcommand{\wt}{\widetilde}

\def\E{{\mathbb E}}                         
\def\V{{\mathbb V}}

\def\a{\alpha}

\def\s{\sigma}

\def\t{\theta}
\def\b{\beta}
\def\l{\lambda}

\def\N{{\cal N}}

\def\Kl{\text{\rm KL}}

\def\vec{\text{\rm vec}}
\def\vech{\text{\rm vech}}

\def\cov{\text{\rm cov}}
\def\IS{\text{\rm IS}}

\begin{document}
\doublespacing
\title{Variational Bayes with Intractable Likelihood}
\author{\normalsize Minh-Ngoc Tran, David J. Nott, Robert Kohn\footnote{  
Minh-Ngoc Tran is Lecturer, The University of Sydney Business School, Sydney 2006 Australia
(minh-ngoc.tran@sydney.edu.au). 
David J. Nott is Associate Professor, Department of Statistics and Applied Probability,
National University of Singapore, Singapore 117546 (standj@nus.edu.sg).
Robert Kohn is Professor, UNSW Business School,
University of New South Wales, Sydney 2052 Australia (r.kohn@unsw.edu.au).}}
\date{\empty}
\maketitle
\begin{abstract}
Variational Bayes (VB) is rapidly becoming a popular tool for Bayesian inference in statistical modeling.
However, the existing VB algorithms are restricted to cases where the likelihood is tractable,
which precludes the use of VB in many interesting situations such as in state space models and in approximate Bayesian computation (ABC),
where application of VB methods was previously impossible.   
This paper extends the scope of application of VB to cases where the likelihood is intractable, 
but can be estimated unbiasedly.
The proposed VB method therefore makes it possible to carry out Bayesian inference in many statistical applications,
including state space models and ABC.
The method is generic in the sense that it can be applied to almost all statistical models without requiring too much model-based derivation,
which is a drawback of many existing VB algorithms.
We also show how the proposed method can be used to obtain highly accurate VB approximations of marginal posterior distributions. 

\paradot{Keywords} Approximate Bayesian computation, marginal likelihood, natural gradient, quasi-Monte Carlo, state space models, stochastic optimization. 

\end{abstract}
\section{Introduction}\label{sec:Introduction}
Let $y$ be the data and $\t\in\Theta$ the vector of model parameters.
We are interested in Bayesian inference about $\t$, based on the posterior distribution $\pi(\t)=p(\t|y)\propto p(\t)p(y|\t)$, with $p(\t)$ the prior
and $p(y|\t)$ the likelihood function.
In this paper, we are interested in Variational Bayes (VB), which is widely used as a computationally effective method
for approximating the posterior distribution $\pi(\t)$ \citep{Attias:1999,Bishop:2006}.
VB approximates the posterior by a distribution $q(\t)$ within some tractable class, such as an exponential family, chosen to minimize the 
Kullback-Leibler divergence between $q(\t)$ and $\pi(\t)$.
Most of the VB algorithms in the literature require that
the likelihood $p(y|\t)$ can be computed analytically for any $\t$.

In many applications, however, the likelihood $p(y|\t)$ is intractable
in the sense that it is infeasible to compute $p(y|\t)$ exactly at each value of $\t$,
which makes it difficult to use VB for inference.
For example, in state space models \citep{Durbin:2001}, which are widely used in economics,
finance and engineering, the likelihood is a high dimensional integral over the state variables governed by a Markov process.
\cite{Ghahramani:2000} were the first to use VB for inference in state space models.
However, they only consider the special case in which the time series is segmented into regimes with each regime assumed to follow a linear-Gaussian state space model.
For general state space models, it is still a challenging problem to do inference with VB. 
\cite{Turner:2011} discuss some of the difficulties in applying VB methods to time series models.
Another example of a situation where implementing VB is difficult is approximate Bayesian computation (ABC) 
\citep{Tavare:1997,Marin:2012,Peters:2012}.
ABC methods provide a way of approximating the 
posterior $\pi(\t)$ when the likelihood is difficult to compute but it is possible to simulate data from the model.
We are not aware of any work that uses VB for inference in ABC,
although a closely related technique called Expectation Propagation has been used \citep{Barthelme:2014}.
This paper proposes a VB algorithm that approximates $\pi(\t)$ when the likelihood is intractable.
The only requirement is that the intractable likelihood can be estimated unbiasedly.
The proposed algorithm therefore makes it possible to carry out variational Bayes inference in many statistical models with an intractable likelihood, 
where this was previously impossible.

In many models, by introducing a latent variable $\a$, the joint density $p(y,\a|\t)$ is tractable.
This makes it much easier to work with the joint posterior $p(\t,\a|y)\propto p(\t)p(y,\a|\t)$ rather than the marginal posterior of interest $\pi(\t)$ itself.
In this situation many VB algorithms in the literature approximate 
the joint posterior $p(\t,\a|y)$ by a factorized distribution $q(\t)q(\a)$,
and then use $q(\t)$ as an approximation to $\pi(\t)$.
The main drawback of this approach is that the (usually high) posterior dependence between $\t$ and $\a$ is ignored, which might lead to a poor VB approximation 
\citep{Neville:2014}. Our VB algorithm approximates $\pi(\t)$ directly
with the latent variable $\a$ integrated out and thus overcomes this drawback; see the example in Section \ref{sec:GLMM}.

Section \ref{sec:VBIL} presents our approach, which we call Variational Bayes with Intractable Likelihood (VBIL),
when the likelihood can be estimated unbiasedly.
VBIL transforms the problem of approximating the posterior $\pi(\t)$
into a stochastic optimization problem using a noisy gradient.
It is essential for the success of stochastic optimization algorithms
to have a gradient estimator with a sufficiently small variance. 
Section \ref{sec:variation reduction} describes several techniques, including control variate and quasi-Monte Carlo, 
for variance reduction in estimating the gradient.
This section also discusses the importance of the natural gradient \citep{Amira:1998},
which takes into account the geometry of the variational distribution $q(\t)$ being learned.

Unlike many VB algorithms that are derived on a model-by-model basis and require analytical computation of some model-based expectations,
one of the main advantages of VBIL is that it can be applied to almost all statistical models without requiring an analytical solution to model-based expectations.
The only requirement is that we are able to estimate the intractable likelihood unbiasedly.
The VBIL methodology is therefore generic and widely applicable.
As a by-product, VBIL provides an estimate of the marginal likelihood, which is useful for model choice.

There are several lines of work related to ours in terms of working with an intractable likelihood.
\cite{Beaumont:2003} and \cite{Andrieu:2009} show that Markov chain Monte Carlo simulation 
based on an unbiased estimator of the likelihood is still able to generate samples from the
posterior. This method is known in the literature as Pseudo-Marginal Metropolis-Hasting (PMMH).   
More efficient variants of PMMH, called correlated PMMH and blockwise PMMH, have been proposed recently \citep{Deligiannidis:2015,Tran:2016}.
\cite{Tran:2013} show that importance sampling with the likelihood replaced by its unbiased estimator is still valid for estimating expectations with respect to the posterior, and name their method as Importance Sampling Squared ($\IS^2$).
Also, \cite{Duan:2013} and \cite{Tran:2014} use sequential Monte Carlo for inference based on an unbiased likelihood estimator.
The main advantage of VBIL is that it is several orders of magnitude faster than these competitors.

Section \ref{sec:optimal sigma} studies the link between
the precision of the likelihood estimator to the variance of the VBIL estimator.
This helps to understand how much accuracy is lost when working with an estimated likelihood
compared to the case the likelihood is available.
In this spirit, \cite{Pitt:2012} and \cite{Tran:2013} show that the asymptotic variance of PMMH and $\IS^2$ estimators
increases exponentially with the variance of the likelihood estimator.
Therefore, it is critical for these methods to have a likelihood estimator that is accurate enough.
They show that the variance of the likelihood estimator should be around 1 in order to minimize
the computing time that is needed for the variance of PMMH/$\IS^2$ estimators to have a fixed precision.
For VBIL, we show that the asymptotic variance of VBIL estimators
increases linearly with the variance of the likelihood estimator.
The proposed methodology is therefore useful in cases when only highly
variable estimates of the likelihood are available.
We discuss such a situation in Section \ref{sec:GLMM}
where VBIL works well while its competitors fail.

Several interesting applications of VBIL are presented in Section \ref{sec:applications}.
Section \ref{sec:GLMM} shows the use of VBIL for generalized linear mixed models
and demonstrates the high accuracy of VBIL compared to the existing VB algorithms.
Section \ref{sec:ss models} applies VBIL to Bayesian inference in state space models
and Section \ref{sec:ABC} shows how VBIL can be used for ABC. 
To the best of our knowledge, our paper is the first
to use a VB method in the most general way for Bayesian inference in state space models and ABC.
Another interesting application of VBIL is presented in Section \ref{sec:improve VB},
in which we illustrate that VBIL provides an attractive way to improve the accuracy of VB approximations of marginal posteriors.

\section{Variational Bayes with an intractable likelihood}\label{sec:VBIL}
This section describes the basic form of the proposed VBIL algorithm,
where an unbiased estimator of the likelihood is available.
Denote by $\wh p_N(y|\t)$ an unbiased estimator of the likelihood $p(y|\t)$.
Here $N$ is an algorithmic parameter relating to the precision in estimating the likelihood,
such as the number of samples if the likelihood is estimated by importance sampling
or the number of particles if the likelihood in state space models is estimated by a particle filter.
Using the terminology in \cite{Pitt:2012}, we refer to $N$ as the number of particles. 
Let $z=\log\wh p_N(y|\t)-\log p(y|\t)$, so that $\wh p_N(y|\t)=p(y|\t)e^z$, and denote by $g_N(z|\t)$ the density of $z$. 
Note that $z$ is unknown as we do not know $\log p(y|\t)$ and there is no need to compute $z$ in practice,
but, as will become clear shortly, it is very convenient to work with $z$.
We sometimes write $\wh p_N(y|\t)$ as $\wh p_N(y|\t,z)$.
We note that $\int e^zg_N(z|\t)dz=1$ because of the unbiasedness of the estimator $\wh p_N(y|\t)$.
Define the following density on the extended space $\Theta\times\Bbb{R}$
\beqn
\pi_N(\t,z)=\frac{p(\t)p(y|\t)e^zg_N(z|\t)}{p(y)}=\pi(\t)e^zg_N(z|\t).
\eeqn
This augmented density admits the posterior of interest $\pi(\t)$ as its marginal. 
It is useful to work with $\pi_N(\t,z)$ as 
the high-dimensional vector of random variables involved in estimating the likelihood is transformed into the scalar $z$.
A direct approximation of $\pi_N(\t,z)$ is $\wt q_{\l,N}(\t,z)=q_\lambda(\t)e^zg_N(z|\t)$,
where $q_\lambda(\t)$ is the variational distribution with the variational parameter $\lambda$ to be estimated, 
and then $q_\lambda(\t)$ can be used as an approximation of $\pi(\t)$.
However, it turns out that it is impossible to estimate the gradient of 
the \KL{} divergence between $\wt q_{\l,N}(\t,z)$ and $\pi_N(\t,z)$ as this requires knowing $z$.

We propose instead to approximate $\pi_N(\t,z)$ by $q_{\l,N}(\t,z)=q_\lambda(\t)g_N(z|\t)$.
This augmented density has the attractive features that $q_\l(\t)$ is its marginal for $\t$
and it is possible to estimate the gradient of the \KL{} divergence KL($\lambda$) between $q_{\l,N}(\t,z)$ and $\pi_N(\t,z)$ (c.f. \eqref{eq:derivative KL} below).
Although $q_{\l,N}(\t,z)$ does not provide a good approximation of the posterior marginal of $z$, the latter is not of interest to us.
Furthermore, under Assumptions 1 and 2 given in Section 4,
the minimization of KL($\lambda$) is equivalent to the minimization of the KL divergence between $q_\l(\t)$ and $\pi(\theta)$. 

The \KL{} divergence between $q_{\l,N}(\t,z)$ and $\pi_N(\t,z)$ is 
\beq\label{eq:KL lambda}
\text{KL}(\l)=\int q_\lambda(\t)g_N(z|\t)\log\frac{q_\lambda(\t)g_N(z|\t)}{\pi_N(\t,z)}dzd\t,
\eeq
where we omit to indicate dependence on $N$ for notational convenience.
The gradient of $\text{KL}(\l)$ is
\bea\label{eq:derivative KL}
\nabla_\lambda\text{KL}(\l)&=&\nabla_\lambda\int q_\lambda(\t)g_N(z|\t)\log\frac{q_\lambda(\t)}{p(\t)\wh p_N(y|\t,z)}dzd\t\notag\\
&=&\int\left(\nabla_\lambda[q_\lambda(\t)]g_N(z|\t)\log\frac{q_\lambda(\t)}{p(\t)\wh p_N(y|\t,z)}+q_\lambda(\t)g_N(z|\t)\nabla_\lambda[\log q_\lambda(\t)]\right)dzd\t\notag\\
&=&\int\left(q_\lambda(\t)g_N(z|\t)\nabla_\lambda[\log q_\lambda(\t)]\Big(\log q_\lambda(\t)-\log(p(\t)\wh p_N(y|\t,z)\Big)\right)dzd\t\notag\\
&=&\E_{\t\sim q_\l(\t),z\sim g_N(z|\t)}\left(\nabla_\lambda[\log q_\lambda(\t)]\Big(\log q_\lambda(\t)-\log(p(\t)\wh p_N(y|\t,z)\Big)\right).
\eea
Here, we have used the facts that $\nabla_\lambda[q_\lambda(\t)]=q_\lambda(\t)\nabla_\lambda[\log q_\lambda(\t)]$ and that $\E(\nabla_\lambda[\log q_\lambda(\t)])=0$.
It follows from \eqref{eq:derivative KL} that, by generating $\t\sim q_\lambda(\t)$ and $z\sim g_N(z|\t)$,  it is straightforward to obtain an unbiased estimator $\wh{\nabla_\l\text{KL}}(\l)$ of the gradient $\nabla_\l\text{KL}(\l)$.
Therefore, we can use stochastic optimization to optimize $\text{KL}(\lambda)$.
We note that the unknown $z$ is implicitly generated when the unbiased likelihood estimator $\wh p_N(y|\t)=\wh p_N(y|\t,z)$ is computed. In practice, $z$ is never dealt with explicitly and it only plays
a theoretical role in the mathematical derivations.
The basic algorithm is as follows
\begin{algorithm}\label{algorithm 1}
\begin{itemize}
  \item Initialize $\l^{(0)}$ and stop the following iteration if the stopping criterion is met.
  \item For $t=0,1,...$, compute $\l^{(t+1)}=\l^{(t)}-a_t\wh{\nabla_\l\text{KL}}(\l^{(t)})$. 
\end{itemize}
\end{algorithm}
We will refer to this algorithm as Variational Bayes with Intractable Likelihood (VBIL).
The sequence $\{a_t\}$ should satisfy $a_t>0$, $\sum_t a_t=\infty$ and $\sum_t a_t^2<\infty$.
We choose $a_t=1/(1+t)$ in this paper.
It is also possible to train $a_t$ adaptively.

It is important to note that each iteration is parallelizable,
as the gradient ${\nabla_\l\text{KL}}(\l)$ is estimated by independent samples
from $q_{\l,N}(\t,z)$. 

\subsection{Stopping criterion and marginal likelihood estimation} 
An easy-to-implement stopping rule is to stop the updating
procedure if the change between $\lambda^{(t+1)}$ and $\lambda^{(t)}$, e.g. in terms of the Euclidean distance, is less than some threshold $\epsilon$ \citep{Ranganath:2014}.
However, it is difficult to select $\epsilon$ as such a distance depends on the scales and the length of the vector $\l$.
It is easy to show that
\beq
\log p(y) = \int\log\left(\frac{p(\t)\wh p_N(y|\t,z)}{q_\l(\t)}\right)q_{\l,N}(\t,z)dzd\t+\Kl(\l)\geq \text{LB}(\l),
\eeq
where
\bea\label{eq:lb}
\text{LB}(\l) &=& \int\log\left(\frac{p(\t)\wh p_N(y|\t,z)}{q_\l(\t)}\right)q_{\l,N}(\t,z)dzd\t\notag\\
&=&\E_{\t,z}[\log p(\t)-\log q_\l(\t)+\log\wh p_N(y|\t,z) ]
\eea
is the lower bound on the log of the marginal likelihood $\log\;p(y)$.
This lower bound after convergence can be used as an approximation to $\log\;p(y)$,
which is useful for model selection purposes.
The expectation of the first two terms in \eqref{eq:lb} can be computed analytically,
while the last term can be estimated unbiasedly by samples from $q_{\l,N}(\t,z)$.
However, in our experience, estimating the entire expectation \eqref{eq:lb}
based on samples from $q_{\l,N}(\t,z)$ leads to a smaller variance.  
Denote by $\wh{\text{LB}}(\l)$ the resulting unbiased estimate of $\text{LB}(\l)$.
Although $\text{LB}(\l)$ is strictly non-decreasing over iterations,
its sample estimate $\wh{\text{LB}}(\l)$ might not be.
To account for this, we suggest to stop the updating procedure if the change in an averaged value of the lower bounds
over a window of $M$ iterations, $\overline {\text{LB}}(\l_t)=(1/M)\sum_{k=1}^M \wh{\text{LB}}(\l_{t-k+1})$,
is less than some threshold $\epsilon$.
At convergence, the values $\text{LB}(\l_t)$ stay the same,
therefore $\overline {\text{LB}}(\l_t)$ will average out the noise in $\wh{\text{LB}}(\l_t)$ and is stable. 
Furthermore, we suggest to replace $\wh{\text{LB}}(\l)$ by a scaled version of it,
$\wh{\text{LB}}(\l)/n$ with $n$ the size of the dataset such as the number of observations.
The scaled lower bound is more or less independent of the size of the dataset (c.f., Figure \ref{f:lower bound figure}).
We set $M=5$ and $\epsilon=10^{-5}$ in this paper.

\section{Variance reduction and natural gradient}\label{sec:variation reduction}
As is typical of stochastic optimization algorithms,
the performance of Algorithm \ref{algorithm 1} depends greatly on the variance of the noisy gradient.
This section describes several techniques for variance reduction.

\subsection{Control variate}\label{sec:control variate}
Denote $\wh h(\t,z)=\log\;(p(\t)\wh p_N(y|\t,z))$ for notational simplicity.
Let $\t_s\sim q_\l(\t)$ and $z_s\sim g_N(z|\t_s)$, $s=1,...,S$, be $S$ samples from the variational distribution $q_{\l,N}(\t,z)$. A naive estimator of the $i$th element of $\nabla_\lambda\text{KL}(\l)$ is
\beq\label{eq: naive KL estimate}
\wh{\nabla_{\lambda_i}\Kl}(\l)^{\text{naive}}=\frac1S\sum_{s=1}^S\nabla_{\lambda_i}[\log q_\lambda(\t_s)]\big(\log q_\l(\t_s)-\wh h(\t_s,z_s)\big),
\eeq
whose variance is often too large to be useful. 
For any number $c_i$, consider   
\beq\label{eq: reduced var KL estimate}
\wh{\nabla_{\lambda_i}\Kl}(\l)=\frac1S\sum_{s=1}^S\nabla_{\lambda_i}[\log q_\lambda(\t_s)](\log q_\l(\t_s)-\wh h(\t_s,z_s)-c_i),
\eeq
which is still an unbiased estimator of $\nabla_{\lambda_i}\text{KL}(\l)$ since $\E(\nabla_{\lambda}[\log q_\lambda(\t)])=0$,
whose variance can be greatly reduced by an appropriate choice of $c_i$.
Similar ideas are considered in the literature, see \cite{Paisley:2012}, \cite{Nott:2012} and \cite{Ranganath:2014}.
The variance of $\wh{\nabla_{\lambda_i}\Kl}(\l)$ is 
\bean
\V(\wh{\nabla_{\lambda_i}\Kl}(\l))&=&\frac1S\V\Big(\nabla_{\lambda_i}[\log q_\lambda(\t)]\big(\log q_\l(\t)-\wh h(\t,z)-c_i\big)\Big)\\
&=&\frac1S\V\Big(\nabla_{\lambda_i}[\log q_\lambda(\t)]\big(\log q_\l(\t)-\wh h(\t,z)\big)\Big)\\
&&\phantom{cccc}-\frac{2c_i}{S}\cov\Big(\nabla_{\lambda_i}[\log q_\lambda(\t)]\big(\log q_\l(\t)-\wh h(\t,z)\big),\nabla_{\lambda_i}[\log q_\lambda(\t)]\Big)  \\
&&\phantom{cccc}+\frac{c_i^2}{S}\V\Big(\nabla_{\lambda_i}[\log q_\lambda(\t)]\Big).
\eean
The optimal $c_i$ that minimizes this variance is 
\beq\label{eq:optimal c_i}
c_i=\cov\Big(\nabla_{\lambda_i}[\log q_\lambda(\t)]\big(\log q_\l(\t)-\wh h(\t,z)\big),\nabla_{\lambda_i}[\log q_\lambda(\t)]\Big)\Big/\V\Big(\nabla_{\lambda_i}[\log q_\lambda(\t)]\Big).
\eeq
Then
\beqn
\V(\wh{\nabla_{\lambda_i}\Kl}(\l))=\V(\wh{\nabla_{\lambda_i}\Kl}(\l)^{\text{naive}})(1-\rho^2_i)\leq\V(\wh{\nabla_{\lambda_i}\Kl}(\l)^{\text{naive}}),
\eeqn
where $\rho_i$ is the correlation between $\nabla_{\lambda_i}[\log q_\lambda(\t)]\big(\log q_\l(\t)-\wh h(\t,z)\big)$ and $\nabla_{\lambda_i}[\log q_\lambda(\t)]$.
Often, $\rho_i^2$ is very close to 1.

We estimate the numbers $c_i$ by samples $(\t_s,z_s)\sim q_{\l,N}(\t,z)$ as in \eqref{eq:optimal c_i}.
In order to ensure the unbiasedness of the gradient estimator, the samples used to estimate $c_i$ must be independent of the samples used to estimate the gradient.
In practice, the $c_i$ can be updated sequentially as follows.
At iteration $t$, we use the $c_i$ computed in the previous iteration $t-1$, i.e. based on the samples from $q_{\l^{(t-1)},N}(\t,z)$,
to estimate the gradient $\wh{\nabla_\l\text{KL}}(\l^{(t)})$,
which is estimated using new samples from $q_{\l^{(t)},N}(\t,z)$.
We then update the $c_i$ using this new set of samples.
By doing so, the unbiasedness is guaranteed while no extra samples are needed in updating the numbers $c_i$.

The gradient in the form of \eqref{eq:derivative KL} can be written as a sum of two terms,
where the first term $\E_{\t\sim q_\l(\t)}(\nabla_\lambda[\log q_\lambda(\t)]\log q_\lambda(\t))$ can be in most cases computed analytically.
However, as pointed out by a referee, this term should be estimated using the same samples of $\t$ as we do in \eqref{eq: reduced var KL estimate}. Doing so helps to reduce the noises in estimating the gradient. This is because the first term plays the role of a control variate. 
This phenomenon is discussed in detail in \cite{Salimans:2013}.

\subsection{Natural gradient}\label{sec:natural gradient}
Intuitively, a different learning rate should be used for each scale in the gradient vector.
That is, the traditional gradient vector $\nabla_\l\Kl(\l)$ should be multiplied by an appropriate scale matrix.
It is well-known that the traditional gradient defined on the Euclidean space
does not adequately capture the geometry of the variational distribution $q_\l(\t)$ \citep{Amira:1998}.
A small Euclidean distance between $\l$ and $\l'$ does not necessarily mean a small \KL{} divergence between $q_\l(\t)$ and $q_{\l'}(\t)$.
\cite{Amira:1998} defines the natural gradient as 
\beq\label{eq:natural gradient}
\nabla_{\lambda}\Kl(\l)^{\text{natural}} = I_F(\l)^{-1}\nabla_\l\Kl(\l),
\eeq
with $I_F(\l)$ the Fisher information matrix,
and suggests using the natural gradient as an efficient alternative to the traditional gradient. See also \cite{Hoffman:2013}.

If the variational distribution $q_\l(\t)$ has the exponential family form
\beq\label{eq:exponential form}
q_\lambda(\theta)=\exp(T(\t)'\lambda-Z(\lambda)),
\eeq 
with $T(\t)$ the vector of sufficient statistics
and $\lambda$ the vector of natural parameters, 
then $I_F(\l)=\cov_{q_\l}\big(T(\t),T(\t)\big)$ is computed analytically.

The use of the natural gradient in VB algorithms is considered, among others, by \cite{Honkela:2010}, \cite{Hoffman:2013} and \cite{Salimans:2013}.  
We demonstrate the importance of the natural gradient using a simple example
where the likelihood is available.
We consider a model where the data $y_i\sim B(1,\t)$ - the Bernoulli distribution with probability $\theta$.
We generate $n=200$ observations $y_i$ from $B(1,\t=0.3)$ and obtain $k=\sum_iy_i=57$.
We use a uniform prior on $\theta$. Then, the posterior $p(\theta|y)$ is Beta($k+1,n-k+1$).
The variational distribution $q_\l(\t)$ is chosen to be Beta($\alpha,\beta$) which belongs to the exponential family
with the natural parameter $\lambda=(\alpha,\beta)'$. The Fisher information matrix $I_F(\lambda)$ is 
\beqn
I_F(\l) = \begin{bmatrix}
\psi_1(\alpha)-\psi_1(\alpha+\beta)&\psi_1(\alpha+\beta)\\
\psi_1(\alpha+\beta)&\psi_1(\beta)-\psi_1(\alpha+\beta)
\end{bmatrix},
\eeqn 
where $\psi_1(x)=\nabla_{xx}[\log\Gamma(x)]$ is the {\it trigamma} function.
In this simple example, the gradient $\nabla_\lambda\text{KL}(\l)$ in \eqref{eq:derivative KL}
can be computed analytically 
\[\nabla_\lambda\text{KL}(\l)=I_F(\l)\l-H(\l)\]
with 
\beqn
H(\l)=\begin{bmatrix}
k\psi_1(\alpha)-n\psi_1(\alpha+\beta)\\
(n-k)\psi_1(\beta)-n\psi_1(\alpha+\beta)
\end{bmatrix}.
\eeqn 
Using the traditional gradient, the update in Algorithm \ref{algorithm 1} is
\beqn
\l^{(t+1)} = \l^{(t)}-a_t\Big(I_F(\l^{(t)})\l^{(t)}-H(\l^{(t)})\Big).
\eeqn
Using the natural gradient, the update is
\beqn
\l^{(t+1)} = (1-a_t)\l^{(t)}+a_tI_F(\l^{(t)})^{-1}H(\l^{(t)}).
\eeqn
Figure \ref{f:binomial_example} plots the densities of the exact posterior $\pi(\t)$ and the variational distributions $q_\l(\t)$ estimated
by the VBIL using the traditional gradient and the natural gradient, with two different random initializations.
The figure shows that the VBIL algorithm using the natural gradient is superior to using the traditional gradient.
Furthermore, the VBIL algorithm based on the natural gradient is insensitive to the initial $\l^{(0)}$.

\begin{figure}[h]
\centering
\includegraphics[width=1\textwidth,height=.3\textheight]{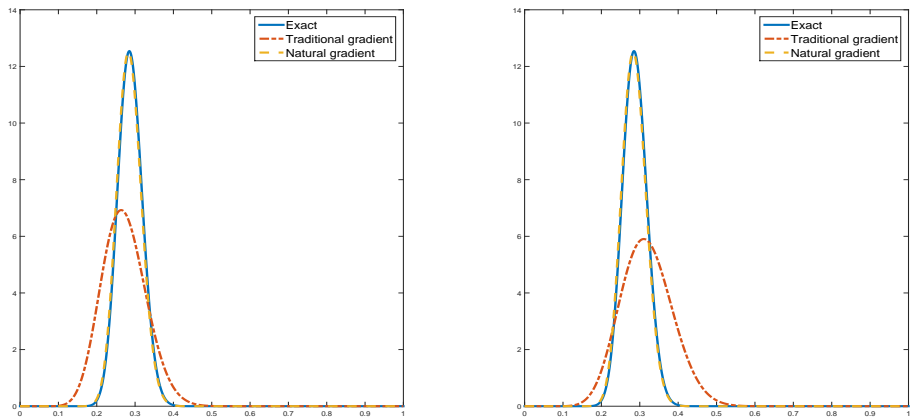}
\caption{Plots of the densities of the exact posterior and the variational approximation estimates, at convergence, with two different starting values $\l^{(0)}$ at random.}
\label{f:binomial_example}
\end{figure}

\subsection{Factorized variational distribution}
Often, the variational distribution $q_\l(\t)$ is factorized into $K$ factors
\beq\label{eq:factorization}
q_\l(\t) = q_{\l^{(1)}}(\t^{(1)})...q_{\l^{(K)}}(\t^{(K)}).
\eeq
Then, each factor $q_{\l^{(k)}}(\t^{(k)})$ is updated separately
and the variance of the estimate of the corresponding gradient can be reduced.
\cite{Salimans:2013} and \cite{Ranganath:2014} consider variance reduction using factorization.
Denote by $\wh h_k(\t,z)$ the terms in $\wh h(\t,z)$ that involve only $\t^{(k)}$ and $z$.
From \eqref{eq:derivative KL}, and noting that $\E_{\t,z}(\nabla_{\l^{(k)}}[\log q_{\l^{(k)}}(\t^{(k)})])=0$,
the traditional gradient corresponding to factor $k$ is
\beq\label{eq:Factorized gra 1}
\nabla_{\l^{(k)}}\text{KL}(\l)=\E_{\t\sim q_\l(\t),z\sim g_N(z|\t)}\left(\nabla_{\l^{(k)}}[\log q_{\l^{(k)}}(\t^{(k)})]\Big(\log q_{\l^{(k)}}(\t^{(k)})-\wh h_k(\t,z)\Big)\right).
\eeq
In the case $q_{\l^{(k)}}(\t^{(k)})=\exp(T_k(\t^{(k)})'\lambda^{(k)}-Z_k(\lambda^{(k)}))$ belongs to an exponential family, 
the natural gradient corresponding to factor $k$ is
\beq\label{eq:Factorized gra 2}
\nabla_{\l^{(k)}}\text{KL}(\l)^\text{natural}=I_{F,k}(\l^{(k)})^{-1}\nabla_{\l^{(k)}}\text{KL}(\l),
\eeq
where $I_{F,k}(\l^{(k)})$ is the Fisher information matrix of distribution $q_{\l^{(k)}}(\t^{(k)})$.

Estimating the gradient using \eqref{eq:Factorized gra 1} has less variation
than using \eqref{eq:derivative KL}.
Intuitively, this is because the variation due to terms not involving $\theta^{(k)}$ has been removed.
This is also explained in \cite{Ranganath:2014} as a Rao-Blackwellization effect.

\subsection{Randomised quasi-Monte Carlo}
Numerical integration using quasi-Monte Carlo (QMC) has been proven efficient in many applications.
Instead of generating uniform random numbers $U(0,1)$ as in plain Monte Carlo methods, 
QMC generates deterministic sequences that are more evenly distributed in (0,1)
in the sense that they minimise the so-called star-discrepancy.
\cite{Dick:2010} provide an extensive background on QMC.
It is shown that, in many cases, QMC integration achieves a better convergence rate than Monte Carlo integration.
In this paper, we use randomised quasi-Monte Carlo (RQMC) as VBIL requires an unbiased estimator of the gradient.
By introducing a random element into a QMC sequence,
RQMC preserves the low-discrepancy property 
and, at the same time, leads to unbiased estimators \citep{Owen:1997,Dick:2010}.

Here, we use RQMC to sample $\t\sim q_\l(\t)$. 
This will help to reduce the variance of the noisy gradient if the dimension of $\t$ is not too high.
Of course, one can also use RQMC in the likelihood estimation, but given some time constraint 
we do not pursue this idea in this paper.

\section{The effect of estimating the likelihood}\label{sec:optimal sigma}
This section studies the effect of the variance of the noisy likelihood
on the VBIL estimators, and provides guidelines for selecting the number of particles $N$.
A large $N$ gives a precise likelihood estimate and therefore an accurate estimate of $\l$,
but at a greater computational cost. A small $N$ leads to a large variance of the
likelihood estimator, so a larger number of iterations is needed for the procedure to settle down.
It is therefore useful in practice to have some guidelines for selecting $N$.

In order to understand the effect of estimating the likelihood,
we follow \cite{Pitt:2012} and make the following assumption.
\begin{assumption} \label{ass: assumption 1}
There is a function $\gamma^2(\t)>0$ such that $\E(z|\t)=-\frac{\gamma^2(\t)}{2N}$
and $\V(z|\t)=\frac{\gamma^2(\t)}{N}$.
\end{assumption}
More precisely, \cite{Pitt:2012} assume further that $z\sim\N(-\frac{\gamma^2(\t)}{2N},\frac{\gamma^2(\t)}{N})$
in order to derive a theory for selecting an optimal $N$.
This assumption is justified in \cite{Tran:2013} and \cite{Doucet:2013} making use of the unbiasedness of the likelihood estimate.
The reason that the mean of $z$ is $-\frac12$ times its variance is because $\E(e^z)=1$ in order for the likelihood estimator to be unbiased. 

\begin{assumption} \label{ass: assumption 2}
For a given $\sigma^2>0$, let $N$ be a function of $\t$ and $\s^2$ such that
$\V(z|\t)\equiv\s^2$, i.e. $N=N_{\s^2}(\t)=\gamma^2(\t)/\s^2$. Then $\E(z|\t)=-\frac{\s^2}{2}$ and $\V(z|\t)=\s^2$.
\end{assumption}

Suppose that the equation $\nabla_\l\Kl(\l)=0$, with $\Kl(\l)$ in \eqref{eq:KL lambda},
has the unique solution $\l^*$. 
Let $\wh \l_M$ be the estimator of $\l^*$ obtained by Algorithm 1 or 2 after $M$ iterations,
and $\wt \l_M$ be the corresponding estimator obtained when the exact likelihood is available.
Denote $\zeta_{*}(\t)=\nabla_\l[\log q_\l(\t)]\big|_{\l=\l^*}$ and denote 
by $\E_*(.)$ and $\V_*(.)$ the expectation and variance operators with respect to $q_{\l^*}(\t)$.
For simplicity, we consider the case that $\l$ is scalar;
the case with a multivariate $\l$ can be obtained using Theorem 5 of \cite{Sacks:1958}.
We obtain the following results whose proof is in the Appendix.

\begin{theorem}\label{eq:effect of sig2} 
Suppose that Assumptions 1 and 2 are satisfied,
and that the regularity conditions in Theorem 1 of \cite{Sacks:1958} hold.\\
(i) Then,
\beq\label{eq:CLT}
\sqrt{M}(\wh \l_{M}-\l^*)\stackrel{d}{\to} \N\Big(0,c_{\l^*}\V\big(\wh{\nabla_{\l}\Kl}(\l^*)\big)\Big),\;\text{as}\;M\to\infty,
\eeq
where $c_{\l^*}$ is a positive constant that depends only on geometric properties of the function $\nabla_{\l}\Kl(\l^*)$
and is independent of the random variables involving in estimating $\nabla_{\l}\Kl(\l^*)$, i.e. $c_{\l^*}$ is independent of $\s^2$.\\  
(ii) Let $\s^2_\text{asym}(\wh \l_M)=c_{\l^*}\V\big(\wh{\nabla_{\l}\Kl}(\l^*)\big)$ be the asymptotic variance of $\wh \l_M$ as $M\to\infty$.
Similarly, let $\s^2_\text{asym}(\wt \l_M)$ be the asymptotic variance of $\wt \l_M$.
Then, 
\beq
\s^2_\text{asym}(\wh \l_M)=\s^2_\text{asym}(\wt \l_M)+\s^2\tau(\l^*,S),
\eeq
where $\tau(\l^*,S)=c_{\l^*}\V_{*}\big\{\zeta_{*}(\t)\big\}/S$ if the noisy traditional gradient is used,
and $\tau(\l^*,S)=c_{\l^*}I_F(\l^*)^{-1}\V_{*}\big\{\zeta_{*}(\t)\big\}I_F(\l^*)^{-1}/S$ if the noisy natural gradient in \eqref{eq:natural gradient} is used.
\end{theorem}
These results show that the variance of VBIL estimators increases linearly with $\s^2$.
For PMMH and $\IS^2$ estimators, \cite{Pitt:2012} and \cite{Tran:2013} show that their variances increase exponentially with $\s^2$.
This means that VBIL is useful in cases where only a rough estimate of the likelihood is available, or it is expensive to obtain an accurate estimate of the likelihood.

We now discuss the issue of selecting $\s^2$.
We note that under Assumption 2, $N$ is tuned depending on $\t$ as $N=N_{\s^2}(\t)=\gamma^2(\t)/\s^2$,
so the time to compute the likelihood estimate $\wh p_N(y|\t)$ is proportional to $1/\s^2$.
Then, \cite{Pitt:2012} and \cite{Tran:2013} show that,
for the PMMH and $\IS^2$ methods, the optimal $\s^2$ that gives an optimal trade-off between 
the CPU time and the variance of the estimators is 1.
For VBIL, the computing time can be defined as
\beq
\text{CT}(\s^2)=\frac{\s^2_\text{asym}(\wh \l_M)}{\s^2}=\frac{\s^2_\text{asym}(\wt \l_M)}{\s^2}+\tau(\l^*,S),
\eeq
where neither $\s^2_\text{asym}(\wt \l_M)$ nor $\tau(\l^*,S)$ depends on $\s^2$.
These results suggest that $\s^2$ should be set to a large value,
as long as it is not too large for the stochastic search procedure in Algorithms 1 and 2 to converge.

\section{Applications}\label{sec:applications}
\subsection{Application to generalized linear mixed models and panel data models}\label{sec:GLMM}
Generalized linear mixed models (GLMM) \citep[see, e.g.][]{Fitzmaurice:2011},
also known as panel data models, use a vector of random effects $\a_i$ to account for the dependence 
between the observations $y_i=\{y_{ij},j=1,...,n_i\}$ measured on the same individual $i$.
Given the random effects $\a_i$, the conditional density $p(y_i|\t,\a_i)$ belongs to an exponential family.
The joint likelihood function of the model parameters $\t$ and the random effects $\a=(\a_1,...,\a_n)$, is tractable
\beqn
p(y,\a|\t)=\prod_{i=1}^np(\a_i|\t)p(y_i|\t,\a_i).
\eeqn     
Typically in the VB literature the joint posterior $p(\t,\a|y)\propto p(\t)p(y,\a|\t)$ is approximated by a variational distribution of the form 
$q(\t)q(\a)$,
and then $q(\t)$ is used as an approximation to the marginal posterior $p(\t|y)$.
For example, \cite{Tan:2013} take this approach but use partially non-centered parameterizations to reduce dependence between parameter blocks.  
\cite{Ormerod:2012} consider frequentist estimation of $\theta$, but using VB methods to integrate out $\alpha$.  
As discussed in the introduction, factorization of the VB distribution generally ignores
the posterior dependence between $\t$ and $\a$, which often leads to underestimating 
the variance in the posterior distribution of $\t$.
Below, we refer to such a VB method as classical VB.

The likelihood, $p(y|\t)=\prod_{i=1}^np(y_i|\t)$ with $p(y_i|\t)=\int p(y_i|\theta,\a_i)p(\a_i|\theta)d\alpha_i$
an integral over the random effects $\a_i$,
is in most cases analytically intractable but can be easily estimated unbiasedly using importance sampling.
Let $h_i(\alpha_i|y,\t)$ be an importance density for $\a_i$.
The integral $p(y_i|\t)$ is estimated unbiasedly by
\beq\label{e:IS_estimator}
\wh p_{N_i}(y_i|\t)=\frac{1}{N_i}\sum_{j=1}^{N_i}w_i(\a_i^{(j)},\t),\;\;w_i(\a_i^{(j)},\t)=\frac{p(y_i|\a_i^{(j)},\t)p(\a_i^{(j)}|\t)}{h_i(\a_i^{(j)}|y,\t)},\;\;\a_i^{(j)}\stackrel{iid}{\sim}h_i(\cdot|y,\t).
\eeq
It is possible to use different $N_i$ for each $p(y_i|\t)$.
Hence, $\widehat{p}_{N}(y|\theta) =\prod_{i=1}^{n}\widehat{p}_{N_i}(y_{i}|\theta)$ 
is an unbiased estimator of the likelihood $p(y|\t)$.
The variance of $z=\log\;\wh p_N(y|\t)-\log\;p(y|\t)$ is
\beq\label{eq:var_z}
\V(z|\t) = \V(\log\widehat{p}_{N}(y|\theta))=\sum_{i=1}^n\V(\log\wh p_{N_i}(y_i|\t)),
\eeq
which can be estimated by $\wh\V(z|\t) = \sum_{i=1}^n\wh\V(\log\;\wh p_{N_i}(y_i|\t))$ with
\beq\label{eq:var_z_i estimate}
\wh\V(\log\wh p_{N_i}(y_i|\t))=\frac{\wh\gamma_i(\t)}{N_i},\;\;\wh\gamma_i(\t) = \frac{N_i\sum_{j=1}^{N_i}w_i(\a_i^{(j)},\t)^2}{\big(\sum_{j=1}^{N_i}w_i(\a_i^{(j)},\t)\big)^2}-1.
\eeq
Given a fixed $\s^2$, it is therefore straightforward to target $\V(z|\t)=\s^2$ by selecting 
$N_i$ such that $\wh\V(\log\;\wh p_{N_i}(y_i|\t))\approx \s^2/n$.
 
\subsubsection*{Six City data}
We now illustrate the VBIL algorithm using the Six City data in \cite{Fitzmaurice:1993}.
The data consist of binary responses $y_{ij}$ which indicate the wheezing status (1 if wheezing, 0
if not wheezing) of the $i$th child at time-point $j$, $i = 1, . . . , 537$ and $j = 1,...,4$.
Covariates are the age of the child at time-point $j$, centered at 9 years, and the maternal smoking status (0 or 1).
We consider the following logistic regression model with a random intercept
\begin{eqnarray*}
p(y_{ij}|\beta,\a_i)&=&\text{Binomial}(1,p_{ij}),\\
\text{logit}(p_{ij})&=&\beta_1+\beta_2\text{Age}_{ij}+\beta_3\text{Smoke}_{ij}+\a_{i},\;\;\;\a_i\sim \N(0,\tau^2).
\end{eqnarray*}
The model parameters are $\t=(\beta,\tau^2)$.
We use a normal prior $\N(0,50I_3)$ for $\beta$ and a Gamma$(1,0.1)$ prior for $\tau^2$. 

\begin{figure}[h]
\centering
\includegraphics[width=1\textwidth,height=.3\textheight]{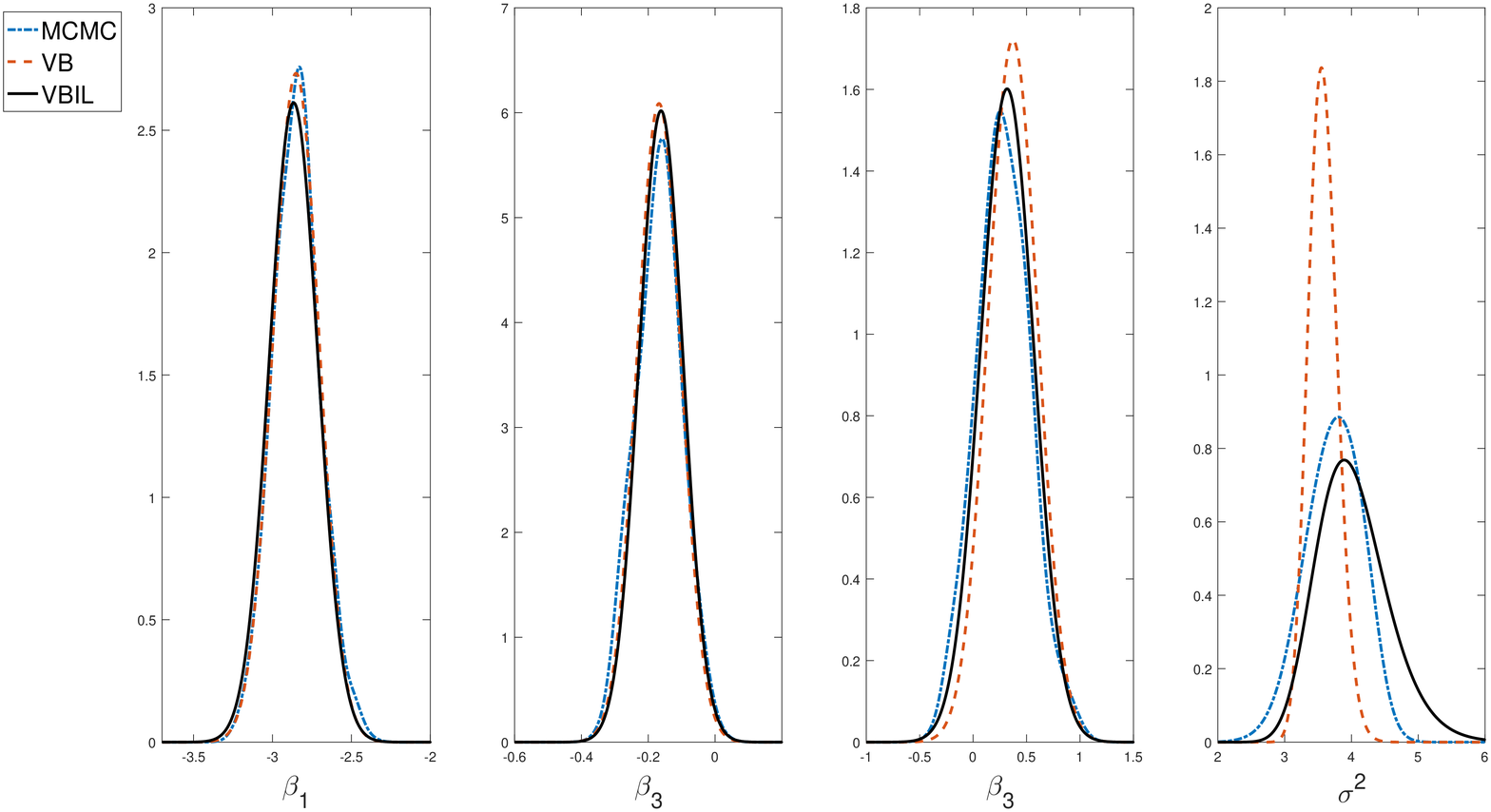}
\caption{Application to GLMM: Six City data}
\label{f:demonstration figure}
\end{figure}

\begin{figure}[h]
\centering
\includegraphics[width=1\textwidth,height=.3\textheight]{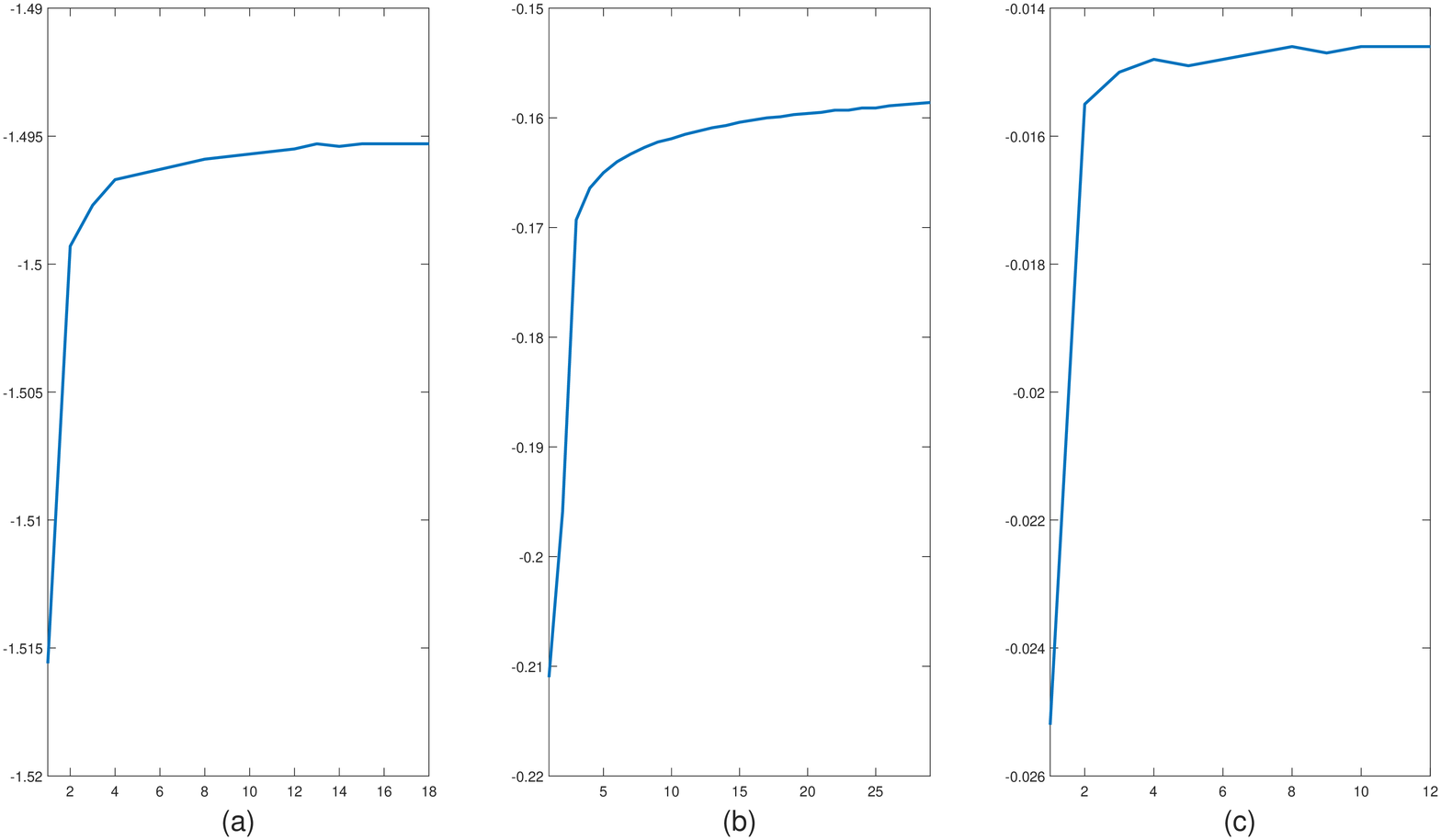}
\caption{Plots of scaled lower bounds over the iterations: (a) Six City example, (b) state space example,
(c) ABC example}
\label{f:lower bound figure}
\end{figure}

We use the variational distribution $q_\l(\t)=q(\beta)q(\tau^2)$,
where $q(\beta)$ is a $d=3-$variate normal $\N(\mu,\Sigma)$ and $q(\tau^2)$ is an inverse gamma distribution.
We then run Algorithm 2, see the Appendix for the detail, with $S=1000$ samples to estimate the gradient. 
The likelihood is estimated as in \eqref{e:IS_estimator} with 
the natural sampler $h_i(\a_i|y,\t)=p(\a_i|\t)$,
which is the normal distribution $\N(0,\tau^2)$ in our case.
The $\s^2$ in Section \ref{sec:optimal sigma} is set to $4$,
which on average requires $\bar N=\sum N_{i}/n=124$ particles.
Using a larger $\s^2$ will lead to too small $N_i$ that makes the estimate in \eqref{eq:var_z_i estimate} unreliable. 
Figure \ref{f:lower bound figure}(a) plots the scaled lower bounds over the iterations.

We compare the performance of the classical VB and VBIL algorithms to the pseudo-marginal MCMC simulation \citep{Andrieu:2009}.
We set $\s^2=1$ as suggested in \cite{Pitt:2012}.
The MCMC chain, based on the adaptive random walk Metropolis-Hastings algorithm in \cite{Haario:2001},
consists of 20000 iterates with another 10000 iterates used as burn-in.

Figure \ref{f:demonstration figure} plots the classical VB estimates (dashed line), 
MCMC estimates (dotted line) and the VBIL estimates (solid line) of the marginal posteriors $p(\beta_i|y)$ and $p(\tau^2|y)$. 
The MCMC density estimates are carried out using the kernel density estimation method based on the built-in Matlab function
\texttt{ksdensity}.
The figure shows that the VBIL estimates are very close to the MCMC estimates.
The classical VB underestimates the posterior variance of $\tau^2$ in this example.
The clock times taken to run the VB, VBIL and MCMC procedures are 4, 2.9 and 505 minutes respectively.
However, we note that the running time depends on many factors
such as the programming language being used and the initialization of the procedures.
All the examples in this paper are run on 
an Intel Core 16 i7 3.2GHz desktop supported by the Matlab Parallel Toolbox with 8 local processors.
Obviously, the more processors we have, the faster the VBIL procerdure is. 

\subsubsection*{Large data example} 
One of the main advantages of VBIL is its scalability,
i.e. it is applicable in large data cases.
This section describes a scenario where it is difficult to use the PMMH and $\IS^2$ methods.
Consider a large data case with a large number of panels $n$.
From \eqref{eq:var_z}, for fixed $N_i$, the variance of the log-likelihood estimator $\V(z|\t)$
increases linearly with $n$.
Therefore, when $n$ is large enough, the PMMH and $\IS^2$ methods will not work in a practical sense, because $\V(z|\t)$ can be very large \citep{Flury:2011}.
In this GLMM setting, PMMH and $\IS^2$ do not work when $\V(z|\t)$ is as large as 6 or 7.
One can decrease $\V(z|\t)$ by increasing $N_i$, but this can be too computationally expensive to be practical.  

We generate a data set of $n=3000$ from the following logistic model with a random intercept 
\bea\label{e:bino sil}
p(y_{ij}|\beta,\a_i)&=&\text{Binomial}(1,p_{ij}),\\
\text{logit}(p_{ij}) &=&\beta_1+\beta_2x_{ij}+\a_i,\;\;\a_i\sim N(0,\tau^2),\;\;i=1,...,n,j=1,...,n_i,\notag 
\eea
with $\beta = (-1.5,2.5)',\ \tau^2=1.5,\ n_i=5$, $x_{ij}\sim U(0,1)$.
It takes, on average across different $\t$, 30 seconds to carry out each likelihood estimation with the numbers of particles $N_i$ tuned to target $\V(z|\t)=1$,
which requires $\bar N=\sum N_{i}/n=3855$ particles.
So if an optimal PMMH procedure was run on our computer to generate a chain of 30000 iterations as done in the Six City data example, it would take 10.4 days.
We now run VBIL with $\s^2$ set to 30, which on average requires $\bar N=\sum N_{i}/n=187$ particles and 0.7 seconds for each likelihood estimation.
The VBIL procedure stopped after 15 iterations with the clock time taken was 23 minutes.
Figure \ref{f:large data figure} plots the variational approximations of the marginal posteriors,
which are bell shaped as expected with a large dataset. 

\begin{figure}[h]
\centering
\includegraphics[width=1\textwidth,height=.3\textheight]{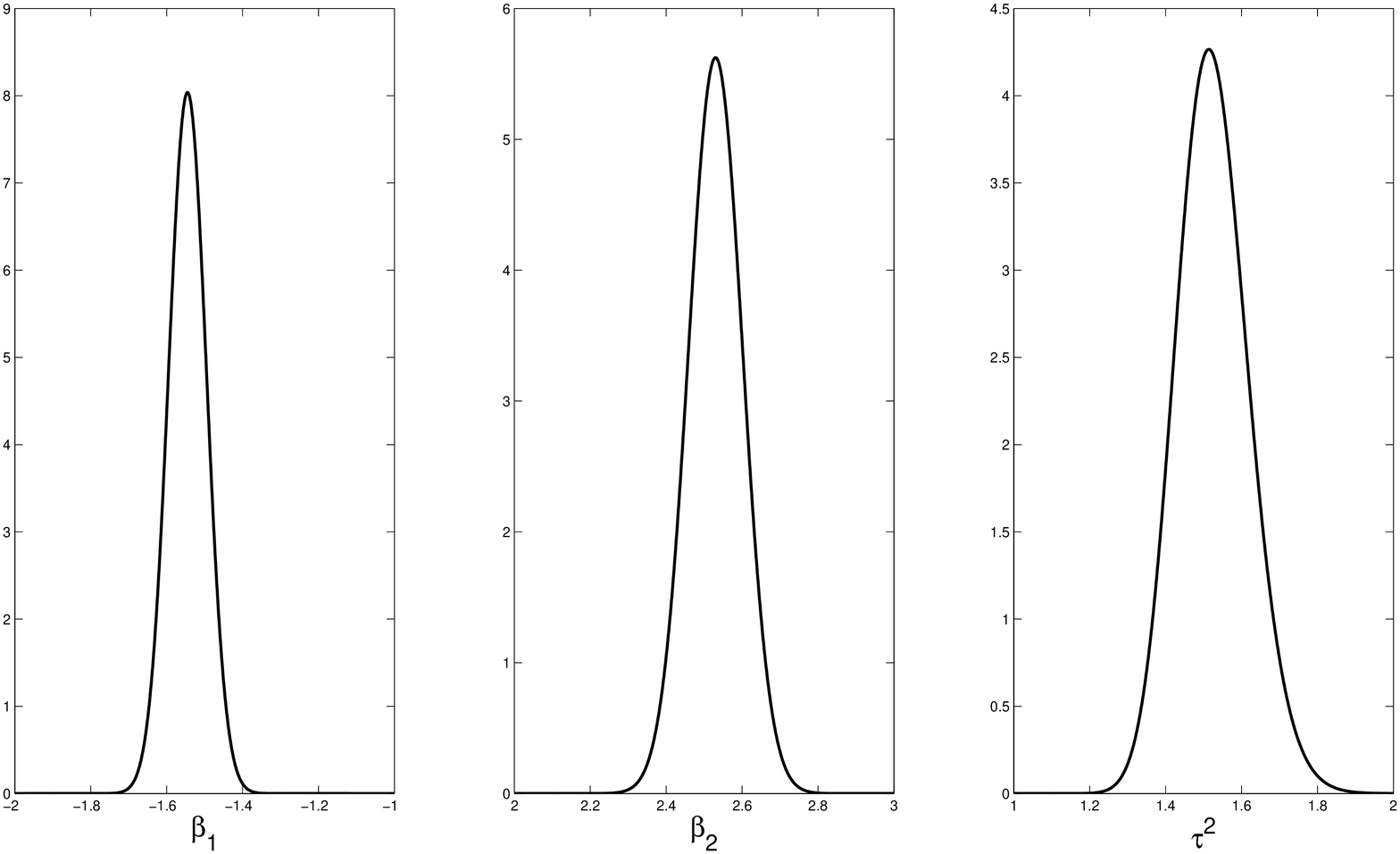}
\caption{Application to GLMM: large data}
\label{f:large data figure}
\end{figure}

\subsection{Application to state space models}\label{sec:ss models}
In state space models, the observations $y_t$ are observed in time order.
At time $t$, the distribution of $y_t$ conditional on a state variable $x_t$ is independently distributed as
\beqn
y_t|x_t\sim g_t(y_t|x_t,\t),
\eeqn
and the state variables $\{x_t\}_{t\geq1}$ are a Markov chain with
\beqn
x_1\sim\mu_\t(\cdot),\;\;\; x_t|x_{t-1}\sim f_t(x_t|x_{t-1},\t).
\eeqn
The likelihood of the data $y=y_{1:T}$ is given by
\beq\label{e:llh}
p(y|\t)=\int p(y|x,\t)p(x|\t)dx
\eeq
with $x=x_{1:T}$ and
\beqn
p(x|\t)=\mu_\t(x_1)\prod_{t=2}^Tf_t(x_t|x_{t-1},\t),\;\;p(y|x,\t)=\prod_{t=1}^Tg_t(y_t|x_t,\t).
\eeqn
Given a value of $\t$, the likelihood $p(y|\t)$ can be unbiasedly estimated by an importance sampling estimator \citep{Shephard:1997,Durbin:1997} or by a particle filter estimator \citep{DelMoral:2004,Pitt:2012}, $\wh p_N(y|\t)$,
with $N$ the number of particles.

An important example of state space models is the stochastic volatility (SV) model.
The time series data $y_t$ is modeled as
\bean
y_t&=&\exp(x_t/2)w_t,\;w_t\sim \N(0,1),\\
x_t&=&\mu+\phi (x_{t-1}-\mu)+\s v_t,\; x_1\sim \N(\mu,\frac{\s^2}{1-\phi^2}),\ v_t\sim \N(0,1),
\eean
with $\mu\in\mathbb{R},\ \phi\in(-1,1)$ and $\sigma^2>0$. Let $\tau=(1+\phi)/2\in(0,1)$;
we will estimate $\tau$ but report results for $\phi$.
The model parameters are $\t=(\mu,\tau,\sigma^2)$. 
We follow \cite{Kim:1998} and use a normal prior $\N(0,10)$ for $\mu$, a Beta prior $B(20,1.5)$ for $\tau$ and an inverse gamma IG(2.5,0.025) for $\s^2$.

To illustrate the VBIL algorithm for state space models,
we analyze the weekday close exchange rates $r_t$ for the Australian Dollar/US Dollar from 5/1/2010 to 31/12/2013.
The data are available from the Reserve Bank of Australia.
The data $y_t$ is
\beqn
y_t=100\left(\log\frac{r_{t+1}}{r_{t}}-\frac{1}{T}\sum_{i=1}^T\log\frac{r_{i+1}}{r_{i}}\right),\;t=1,...,T=1001.
\eeqn
We use the variational distribution $q_\l(\t) = q(\mu)q(\tau)q(\s^2)$,
where $q(\mu)$ is $\N(\mu_\mu,\s^2_\mu)$, $q(\tau)$ is Beta$(\alpha_\tau,\beta_\tau)$ and $q(\s^2)$ is inverse gamma IG$(\alpha_{\s^2},\beta_{\s^2})$.
We employ the constraint $\alpha_\tau>1$ and $\beta_\tau>1$ to make sure that $q(\tau)$ has a mode.
The likelihood estimator $\wh p_N(y|\t)$ is computed by a basic particle filter.
We then run the VBIL algorithm with $S=1000$ samples, starting with 
$\mu_\mu=0,\ \s^2_\mu=0.3$, $\alpha_\tau=95,\ \beta_\tau=5$, $\alpha_{\s^2}=11,\ \beta_{\s^2}=1$.
This initial point is set so that the initial mean values of 
$\mu,\ \phi$ and $\sigma^2$ are $0,\ 0.9$ and $0.1$ respectively,
which is pretty far away from the posterior means; see Figure \ref{f:sv figure}.
The VBIL algorithm stops after 28 iterations.
Figure \ref{f:lower bound figure}(b) plots the scaled lower bounds over the iterations.

The VBIL is compared to pseudo-marginal MCMC simulation,
based on an adaptive random walk Metropolis-Hastings algorithm, 
with 100,000 iterations starting from the same values $\mu=0$, $\tau=0.95$ and $\s^2=0.1$.
The number of particles used in MCMC is fixed at $N=300$, so that
$\V(\wh p_N(y|\bar\t))\approx 1$ at the initial value $\bar\t=(0,0.95,0.1)$.
The number of particles used in VBIL is fixed at $N=100$ as the use of randomised QMC for generating $\t$
helps reduce greatly the variance in estimating the gradient.
We fix $N$ in this example as it is difficult to estimate the variance of log-likelihood estimates
obtained by the particle filter.

Figure \ref{f:sv figure} plots the MCMC
estimates (dotted line) and the VBIL estimates (solid line) of the marginal posteriors.
The figure shows that the VBIL estimates are close to the MCMC estimates
but consume significantly less computational resources.
The CPU times taken to run the VBIL and MCMC procedures are 0.7 and 28 minutes respectively.

\begin{figure}[h]
\centering
\includegraphics[width=1.1\textwidth,height=.4\textheight]{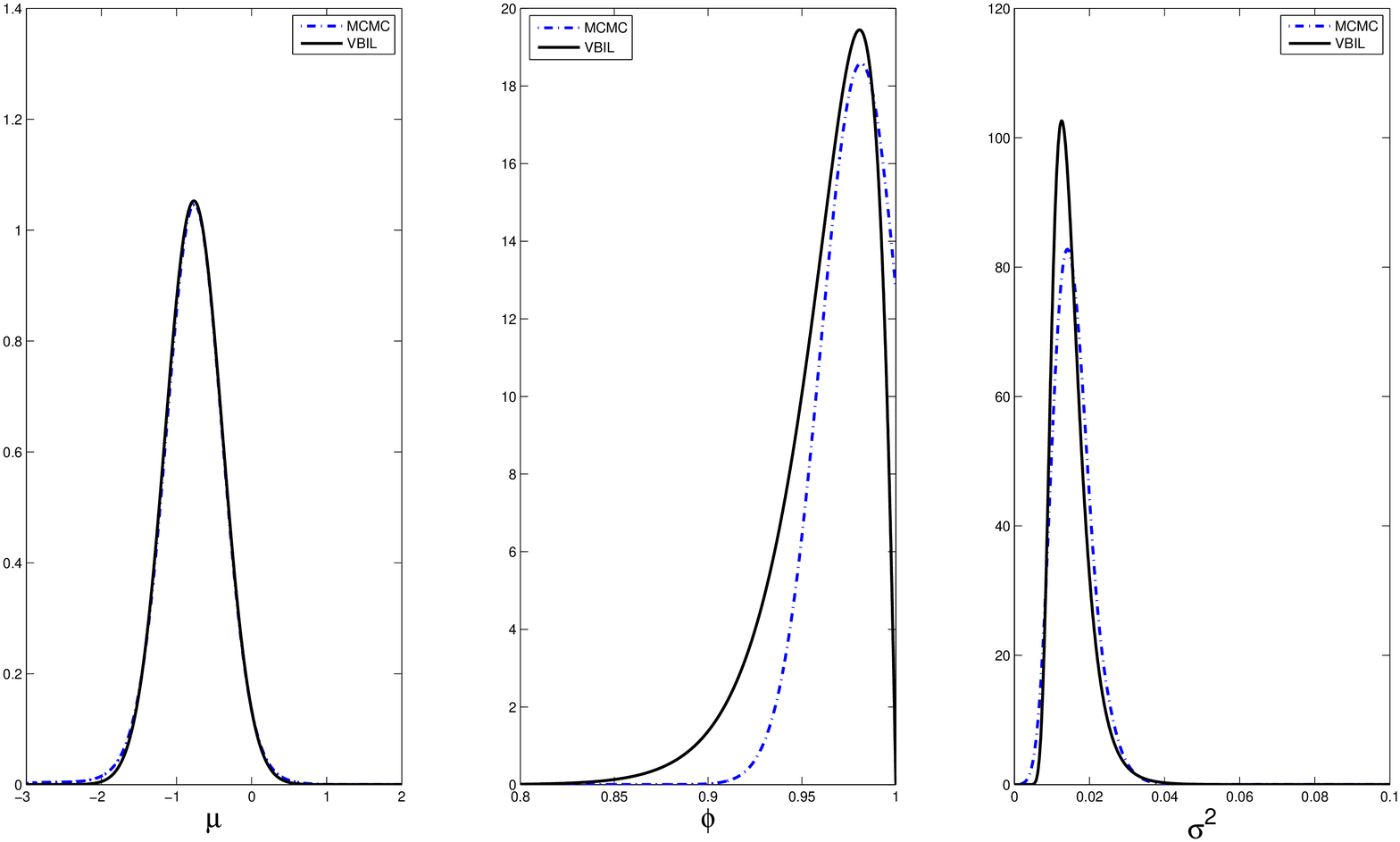}
\caption{Application to state space models: Exchange rate data}
\label{f:sv figure}
\end{figure}
\subsection{Application to ABC}\label{sec:ABC}  
In many modern applications, such as in genetics \citep{Tavare:1997},
we either cannot evaluate the likelihood $p(y|\t)$ pointwise or do not wish to do so, but we can sample from it, i.e. we can simulate $y'\sim p(\cdot|\t)$.   
Approximate Bayesian computation \citep{Tavare:1997} approximates the likelihood by 
\beq\label{eq:likelihood free}
p_\text{LF}(y|\t)=\int K_\epsilon(S(y'),S(y))p(y'|\t)dy',
\eeq
where $K_\epsilon(.,.)$ is a kernel with the bandwidth $\epsilon$
and $S(.)$ is a vector of summary statistics.
Inference is then based on the approximate posterior $p_\text{ABC}(\t|y)\propto p(\t)p_\text{LF}(y|\t)$.
Because the likelihood-free function $p_\text{LF}(y|\t)$ can be unbiasedly estimated by
\beqn
\wh p_N^\text{LF}(y|\t)=\frac{1}{N}\sum_{i=1}^N K_\epsilon(S(y^{[i]}),S(y)),\;\;y^{[i]}\stackrel{iid}{\sim}p(\cdot|\t),
\eeqn 
it is straightforward to use the VBIL algorithm to approximate $p_\text{ABC}(\t|y)$.

We illustrate the application of the VBIL algorithm to ABC by using it to fit an $\a$-stable distribution.
$\alpha$-stable distributions \citep{Nolan:2007} are a class of heavy-tailed distributions 
used in many statistical applications.
An $\alpha$-stable distribution $\mathcal{S}(\alpha,\beta,\gamma,\delta)$ is
parameterized by the stability parameter $\alpha\in (0,2)$, skewness $\beta\in(-1,1)$, scale $\gamma>0$ and location $\delta\in\mathbb{R}$. 
The main difficulty when working with $\alpha$-stable distributions is that they do not have closed form densities,
which makes it difficult to do inference.
However, as it is easy to sample from an $\alpha$-stable distribution,
one can use ABC techniques for Bayesian inference \citep{Peters:2012}.
We illustrate in this example that VBIL provides an efficient approach for fitting $\alpha$-stable distributions.

We generate a data set $y$ with $n=500$ observations 
from a univariate $\alpha$-stable distribution $\mathcal{S}(1.5,0.5,1,0)$.
Let $\wh q_p$ be the $p$th quantile of a pseudo-data set $y'$
generated from $\mathcal{S}(\alpha,\beta,\gamma,\delta)$.
We follow \cite{Peters:2012} and use the summary statistics $S(y')=(\wh v_\alpha,\wh v_\beta,\wh v_\gamma,\wh v_\delta)$ with
\beqn
\wh v_\alpha = \frac{\wh q_{0.95}-\wh q_{0.05}}{\wh q_{0.75}-\wh q_{0.25}},\; 
\wh v_\beta = \frac{\wh q_{0.95}+\wh q_{0.05}-2\wh q_{0.5}}{\wh q_{0.95}-\wh q_{0.05}},\; 
\wh v_\gamma = \frac{\wh q_{0.75}-\wh q_{0.25}}{\gamma},\; 
\wh v_\delta = \frac{1}{n}\sum_{i=1}^ny_i'.
\eeqn
For the observed data $y$, the parameter $\gamma$ in $\wh v_\gamma$
is estimated using McCulloch's method \citep{McCulloch:1986}. 
As the parameterization is discontinuous at $\alpha = 1$, resulting in poor estimates of the summary statistics,
we consider the case with $\a>1$ and restrict the support of $\alpha$ to the interval $(1.1,2)$ as in \cite{Peters:2012}.

We reparameterize 
\beqn
\wt\alpha = \log(\frac{\a-1.1}{2-\a})\in\mathbb{R},\;\;
\wt\beta= \log(\frac{\beta+1}{1-\beta})\in\mathbb{R},\;\;
\wt\gamma=\log(\gamma)\in\mathbb{R},\;\;\wt\delta=\delta\in\mathbb{R},
\eeqn
and estimate $\wt\theta=(\wt\alpha,\wt\beta,\wt\gamma,\wt\delta)$ but report the results for $(\alpha,\beta,\gamma,\delta)$.
We use a normal prior $\wt\t\sim\N(0,100I_4)$ and approximate the posterior $p(\wt\t|y)$ by a normal variational distribution $q_\l(\wt\t) = N(\mu_{\wt\t},\Sigma_{\wt\t})$.
One can work with the original parameterization $(\alpha,\beta,\gamma,\delta)$
and use some form of factorization $q(\a)q(\beta)q(\gamma)q(\delta)$.
We choose to work with $\wt\theta$ to account for the posterior dependence between the parameters.
This also illustrates the flexibility of the VBIL method
in the sense that it can be used without requiring factorization. 

We use the Gaussian kernel with covariance matrix $0.01I_4$ for the likelihood-free $p_\text{LF}(y|\t)$ in \eqref{eq:likelihood free}.
The VBIL is compared to pseudo-marginal Metropolis-Hastings methods with 20,000 iterations after 5000 burnins.
For the standard PMMH \citep{Andrieu:2009}, the number of pseudo-data sets $N=20$ is selected set after some trials in order to have a well-mixing chain.
Efficient versions of PMMH has been proposed recently,
which are more tolerant of noise in the likelihood estimates.
Here we compare VBIL to the blockwise PMMH method of \cite{Tran:2016}. 
For the blockwise PMMH, we set $N=5$.
We also use this value of $N$ in VBIL.
Table \ref{tab: abc example} shows the VBIL and MCMC estimates, and the CPU times.
As shown, VBIL is orders of magnitude faster than MCMC in this example.
Figure \ref{f:lower bound figure}(c) plots the scaled lower bounds over the iterations.

\begin{table}[h]
\centering
\begin{tabular}{c|c|c|c|c}
&True&Standard PMMH &Blockwise PMMH&VBIL\\
\hline
$\alpha$&1.5&1.57 (0.15)&1.58 (0.14)&1.57 (0.11)\\
$\beta$&0.5&0.46 (0.21)&0.45 (0.21)&0.48 (0.16)    \\
$\gamma$&1&1.04 (0.12)&1.04 (0.12)& 1.02 (0.12)    \\
$\delta$&0&-0.08 (0.21)&-0.09 (0.18)&-0.08 (0.14)\\
\hline
CPU time (min)&&12.56&7.62&0.12\\
\end{tabular}
\caption{ABC example: Standard PMMH, blockwise PMMH and VBIL estimates of $\alpha$, $\beta$, $\gamma$ and $\delta$. The numbers in brackets are estimates of the posterior standard deviations.}
\label{tab: abc example}
\end{table}
              

\subsection{Using VBIL to improve marginal posterior estimates}\label{sec:improve VB} 
A drawback of VB methods in general is that the factorization assumption as in \eqref{eq:factorization} 
ignores the posterior dependence between the factors,
which might lead to poor approximations of the posterior variances \citep{Neville:2014}.
We now show how the VBIL algorithm can be used to help overcome this problem.

Suppose that we would like to have a highly accurate VB approximation to the marginal posterior $p(\theta^{(j)}|y)$.
We restrict ourselves to the case with a tractable likelihood for simplicity,
but the following discussion also applies when the likelihood is intractable.
The likelihood of $\theta^{(j)}$, 
\beq\label{eq:likelihood of thetaj}
p(y|\theta^{(j)}) = \int p(\theta^{(\setminus j)}|\theta^{(j)}) p(y|\theta^{(1)},...,\theta^{(K)}) d\theta^{(\setminus j)},
\eeq
with $\theta^{(\setminus j)}=(\theta^{(1)},...,\theta^{(j-1)},\theta^{(j+1)},...,\theta^{(K)})$,
is in general intractable but can be estimated unbiasedly. 
Let $q(\t^{(\setminus j)})$ be an approximation to the marginal posterior $p(\t^{(\setminus j)}|y)$ resulting from a classical VB method
that uses the factorization \eqref{eq:factorization}.
The integral in \eqref{eq:likelihood of thetaj} can be estimated unbiasedly 
using importance sampling with the proposal density $q(\t^{(\setminus j)})$ or a tail-flattened version of it.
This is accurate enough in practice because VBIL does not require a very accurate estimate of $p(y|\theta^{(j)})$ as discussed in Section \ref{sec:optimal sigma}.
The VBIL algorithm can then be used to approximate the marginal posterior $p(\theta^{(j)}|y)$ directly with $\theta^{(\setminus j)}$ integrated out. The resulting approximation is often highly accurate as the dependence between $\theta^{(j)}$ and $\theta^{(\setminus j)}$ is taken into account.
  
A formal justification is as follows. We use the notation as in \eqref{eq:factorization} and write $\l=(\l^{(j)},\l^{(\setminus j)})$.
Suppose that we estimate the marginal posterior of $\l^{(j)}$ by $q_{\l^{(j)}}(\t^{(j)})$ which belongs to a family $\mathcal F=\{q_{\l^{(j)}}(\t^{(j)}),\l^{(j)}\in\Lambda\}$.
VBIL proceeds by minimizing
\beqn
\Kl_j(\l^{(j)})=\int q_{\l^{(j)}}(\t^{(j)})\log\frac{q_{\l^{(j)}}(\t^{(j)})}{p(\t^{(j)}|y)}d\t^{(j)}
\eeqn
over $\l^{(j)}\in\Lambda$. Let $\l^{(j)}_*$ be the VBIL estimator.
Under Assumptions 1 and 2 or when the number of samples $N$ used to estimate \eqref{eq:likelihood of thetaj} is large enough, $\l^{(j)}_*$ is guaranteed to be a minimizer of $\Kl_j(\l^{(j)})$.
Assume further that $\Kl_j(\l^{(j)})$ is convex, then
\beq\label{eq:KL_j inequality}
\Kl_j(\l^{(j)}_*)\leq \Kl_j(\l^{(j)})\;\;\text{for all}\;\;\l^{(j)}\in\Lambda.
\eeq
If we use a VB procedure with a factorization of the form $q_\l(\t)=q_{\l^{(j)}}(\t^{(j)})q_{\l^{(\setminus j)}}(\t^{(\setminus j)})$ 
where $q_{\l^{(j)}}(\t^{(j)})$ belongs to the same family $\mathcal F$,
then VB proceeds by minimizing the KL divergence 
\bea\label{eq:KL factor}
\Kl(\l^{(j)},\l^{(\setminus j)})&=&\int q_{\l^{(j)}}(\t^{(j)})q_{\l^{(\setminus j)}}(\t^{(\setminus j)})\log\frac{q_{\l^{(j)}}(\t^{(j)})q_{\l^{(\setminus j)}}(\t^{(\setminus j)})}{p(\t^{(j)},\t^{(\setminus j)}|y)}d\t^{(j)}d\t^{(\setminus j)}\notag\\
&=&\int q_{\l^{(j)}}(\t^{(j)})\log\frac{q_{\l^{(j)}}(\t^{(j)})}{p(\t^{(j)}|y)}d\t^{(j)}\notag\\
&&\phantom{ccccc}+\int q_{\l^{(j)}}(\t^{(j)})\int q_{\l^{(\setminus j)}}(\t^{(\setminus j)})\log\frac{q_{\l^{(\setminus j)}}(\t^{(\setminus j)})}{p(\t^{(\setminus j)}|\t^{(j)},y)}d\t^{(\setminus j)}d\t^{(j)}\notag\\
&=&\Kl_j(\l^{(j)})+\int q_{\l^{(j)}}(\t^{(j)})\int q_{\l^{(\setminus j)}}(\t^{(\setminus j)})\log\frac{q_{\l^{(\setminus j)}}(\t^{(\setminus j)})}{p(\t^{(\setminus j)}|\t^{(j)},y)}d\t^{(\setminus j)}d\t^{(j)}.
\eea
Let $(\wt\l^{(j)},\wt\l^{(\setminus j)})$ be a minimizer of \eqref{eq:KL factor}. 
Because of the decomposition in \eqref{eq:KL factor}, the estimator $\wt\l^{(j)}$ is not necessarily the minimizer of $\Kl_j(\l^{(j)})$.
From \eqref{eq:KL_j inequality},
\beq\label{eq:VBIL estimator is better}
\Kl_j(\l^{(j)}_*)\leq \Kl_j(\wt\l^{(j)}).
\eeq
So the VBIL estimator $\l^{(j)}_*$ is no worse than the factorization-based VB estimator $\wt\l^{(j)}$ in terms of KL divergence.
   
We illustrate this application by generating $n=100$ observations from a univariate mixture of two normals
\beqn
p(x) = \omega\N(x|\mu_1,\s_1^2)+(1-\omega)\N(x|\mu_2,\s_2^2)
\eeqn   
with $\omega=0.3$, $\mu_1=-3$, $\mu_2=3$, $\s_1^2=2$ and $\s_2^2=3$.
Suppose that we are interested in getting an accurate variational approximation of the posterior $p(\omega|y)$.
Getting an accurate estimate of $w$ is often more challenging than the other parameters.
We use diffuse priors $\omega\sim U(0,1)$, $\mu_1\sim\N(0,100)$, $\mu_2\sim\N(0,100)$, $\sigma_1^2\sim (\sigma_1^2)^{-1}$ and $\sigma_2^2\sim (\sigma_2^2)^{-1}$,
and run VBIL to approximate $p(\omega|y)$ by a Beta distribution.
We use the VB algorithm of \cite{McGrory:2007},
in which the variational distribution is factorized as $q(\omega)q(\s_1^2,\s_2^2)q(\mu_1,\mu_2|\s_1^2,\s_2^2)$,
to design the proposal density to obtain an importance sampling estimator of $p(y|\omega)$. 

Figure \ref{f:improve_VB} plots the McGrory-Titterington estimate (dashed line) and VBIL estimate (solid line) of the posterior $p(\omega|y)$. 
As shown, the VBIL estimate has heavier tails than the VB estimate.
By \eqref{eq:VBIL estimator is better}, it follows that the difference between the two estimates gives an indication of the extent to which the McGrory-Titterington estimate is suboptimal.
This example shows that the VBIL method provides an attractive way to obtain accurate approximation of marginal posteriors.

\begin{figure}[h]
\centering
\includegraphics[width=1\textwidth,height=.3\textheight]{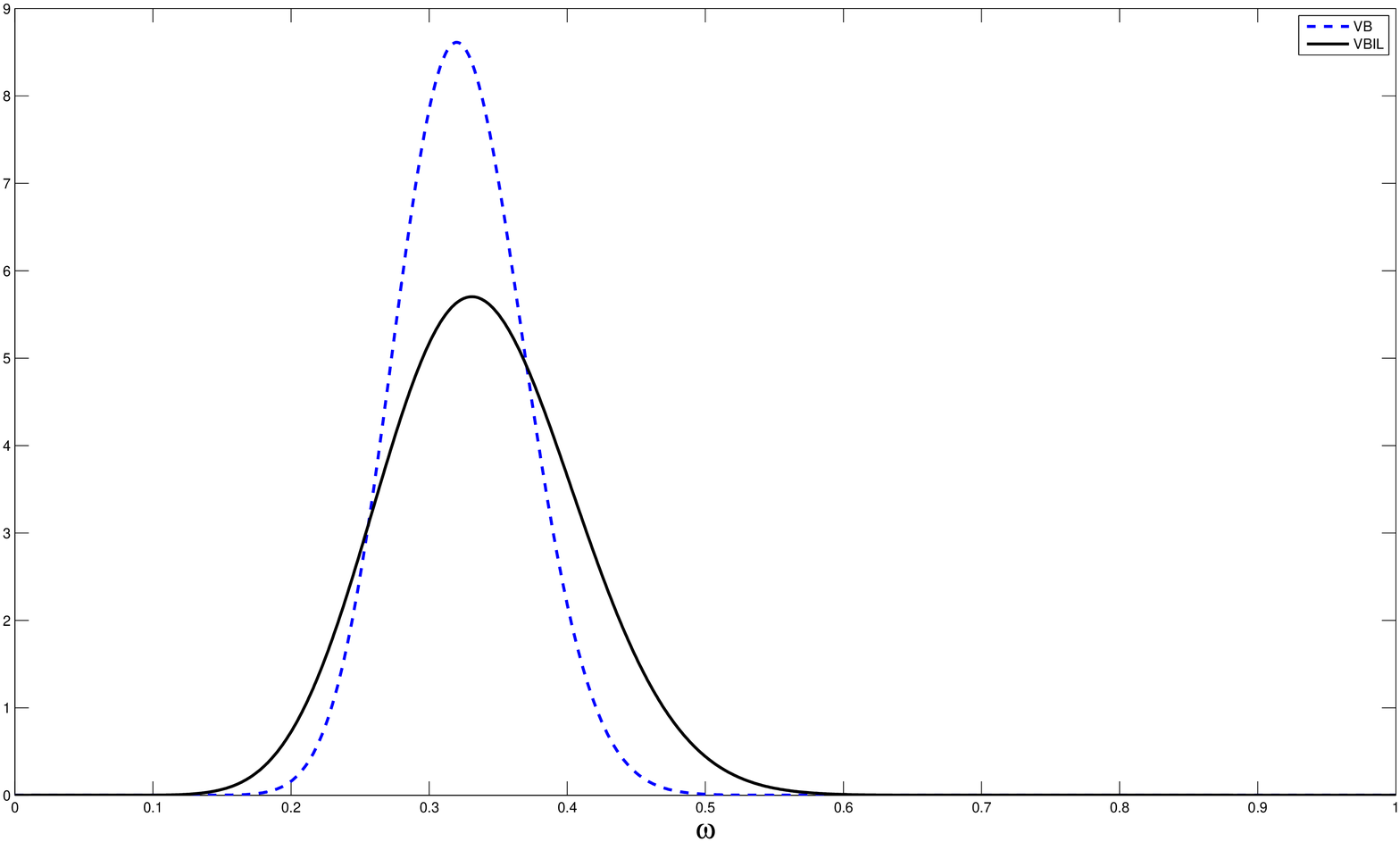}
\caption{Plots the VB (dashed line) and VBIL estimates (solid line) of the posterior $p(\omega|y)$.
}
\label{f:improve_VB}
\end{figure}

\section{Conclusion}\label{sec:conclusion}
We have proposed the VBIL, a useful VB algorithm for Bayesian inference in statistical modeling where the likelihood is intractable.
The method makes it possible to do inference in statistical models
using VB in some situations where that was previously impossible.
The main advantage of VBIL over its competitors, such as PMMH and $\IS^2$, is its scalability.
We show in the examples that VBIL is several orders of magnitude faster than these existing methods.

\subsection*{Acknowledgement}
The research of Tran and Kohn was partially supported by the ARC COE grant CE140100049.
Nott's research was supported by a Singapore Ministry of Education Academic Research Fund Tier 2 grant (R-155-000-143-112).

\section*{Appendix}
\begin{proof}[Proof of Theorem \ref{eq:effect of sig2}]
(i) Under Assumptions 1 and 2, we have that
\beq\label{eq:KL decomposition}
\Kl(\l)=\Kl(q_\l\|\pi)-\int q_\l(\t)\E(z|\t)d\t=\Kl(q_\l\|\pi)+\frac{\s^2}{2},
\eeq
where $\Kl(q_\l\|\pi)$ is the \KL{} divergence between the variational distribution $q_\l(\t)$ and the posterior $\pi(\t)$.
So, $\nabla_{\l}\Kl(\l)=\nabla_{\l}\Kl(q_\l\|\pi)$ is independent of $\s^2$, and minimizing $\Kl(\l)$ with respect to $\l$ is equivalent to minimizing $\Kl(q_\l\|\pi)$.
Algorithm 1 and 2 are the Robbins-Monro procedure for finding the root $\l^*$ of the equation $\nabla_{\l}\Kl(q_\l\|\pi)=0$.
Then, \eqref{eq:CLT} follows from Theorem 1 of \cite{Sacks:1958}
with the constant $c_{\l^*}$ independent of $\s^2$.\\
(ii) Denote $\wh h(\t,z)=\log(p(\t)\wh p_N(y|\t,z))=\log(p(\t)p(y|\t))+z=h(\t)+z$. 
We consider the case with the noisy traditional gradient in \eqref{eq: reduced var KL estimate};
the proof for the other cases is similar. 
We denote by $\wt{\nabla_{\l}\Kl}(\l^*)$ the noisy gradient obtained when the likelihood is available. 
Then, noting that $\E_{*}(\zeta_{*}(\t))=0$,
the constant $c$ in \eqref{eq:optimal c_i} is
\beqn
c=\frac{\E_{\t,z}\{\zeta_{*}(\t)^2(\log q_{\l^*}(\t)-h(\t)-z) \}}{\E_{*}\big\{\zeta_{*}(\t)^2\big\}}=\frac{\E_{*}\{\zeta_{*}(\t)^2(\log q_{\l^*}(\t)-h(\t)) \}}{\E_{*}\big\{\zeta_{*}(\t)^2\big\}}+\frac{\s^2}{2}=\wt c+\frac{\s^2}{2}.
\eeqn
We note that $\wt c$ is the control variate constant we would use to compute 
$\wt{\nabla_{\l}\Kl}(\l^*)$ if the likelihood was known. 
\bean
\V\big(\wh{\nabla_{\l}\Kl}(\l^*)\big)&=&\frac{1}{S}\V_{\t,z}\Big\{\zeta_{*}(\t)(\log q_{\l^*}(\t)-h(\t)-z-c)\Big\}\\
&=&\frac{\s^2}{S}\V_{*}\Big\{\zeta_{*}(\t)\Big\}+\frac{1}{S}\V_{*}\Big\{\zeta_{*}(\t)(\log q_{\l^*}(\t)-h(\t)+\frac{\s^2}{2}-c)\Big\}\\
&=&\frac{\s^2}{S}\V_{*}\Big\{\zeta_{*}(\t)\Big\}+\frac{1}{S}\V_{*}\Big\{\zeta_{*}(\t)(\log q_{\l^*}(\t)-h(\t)-\wt c)\Big\}\\
&=&\frac{\s^2}{S}\V_{*}\Big\{\zeta_{*}(\t)\Big\}+\V\big(\wt{\nabla_{\l}\Kl}(\l^*)\big).
\eean
Therefore,
\bean
\s^2_\text{asym}(\wh \l_M)=c_{\l^*}\V\big(\wh{\nabla_{\l}\Kl}(\l^*)\big)=\s^2_\text{asym}(\wt \l_M)+c_{\l^*}\frac{\s^2}{S}\V_{*}\Big\{\zeta_{*}(\t)\Big\}.
\eean
\end{proof}

\subsection*{Derivation for Section \ref{sec:GLMM}}
The density of the $d-$variate normal $\N(\mu,\Sigma)$ is
\beqn
q(\beta)=\frac{1}{(2\pi)^{d/2}|\Sigma|^{1/2}}\exp\Big(-\frac12(\beta-\mu)'\Sigma^{-1}(\beta-\mu)\Big).
\eeqn
A simplified form of the inverse Fisher matrix for a multivariate normal under the natural parameterization
is given in \cite{Wand:2014}. 
For a $d\times d$ matrix $A$, denote by $\vec(A)$
the $d^2$-vector obtained by stacking the columns of $A$,
by $\vech(A)$ the $\frac12d(d+1)$-vector obtained by stacking the columns of the lower triangular part of $A$.
The duplication matrix of order $d$, $D_d$, is the $d^2\times\frac12d(d+1)$ matrix of zeros and ones such that for any symmetric matrix $A$
\beqn
D_d\vech(A) = \vec(A).
\eeqn 
Let $D_d^{+}=(D_d'D_d)^{-1}D_d'$ be the Moore-Penrose inverse of $D_d$,
and $\vec^{-1}$ be the inverse of the operator $\vec$.
Then, the exponential family form of the normal distribution $q(\b)$ is $q(\beta) = \exp(T(\beta)'\l-Z(\l))$ with
\beq
T(\beta)=\begin{bmatrix}
\beta\\
\vech(\beta\beta')
\end{bmatrix},\;\;\l=\begin{bmatrix}
\l_1\\
\l_2
\end{bmatrix}=\begin{bmatrix}
\Sigma^{-1}\mu\\
-\frac12D_d'\vec(\Sigma^{-1})
\end{bmatrix}.
\eeq
The usual mean and variance parameterization is 
\beqn
\begin{cases}
\mu = -\frac12\big\{\vec^{-1}({D_d^{+}}'\l_2)\big\}^{-1}\l_1\\
\Sigma = -\frac12\big\{\vec^{-1}({D_d^{+}}'\l_2)\big\}^{-1}.
\end{cases}
\eeqn
\cite{Wand:2014} derives the following very useful formula
\beq
I_F(\l)^{-1}=\begin{bmatrix}
\Sigma^{-1}+M'S^{-1}M&-M'S^{-1}\\
-S^{-1}M&S^{-1}
\end{bmatrix},
\eeq
with
\beqn
M=2D_d^{+}(\mu\otimes I_d)\;\;\text{and}\;\; S=2D_d^{+}(\Sigma\otimes\Sigma){D_d^{+}}',
\eeqn
where $\otimes$ is the Kronecker product and $I_d$ the identity matrix of order $d$.
The gradient $\nabla_\l[\log q(\beta)]$ is 
\beq
\nabla_\l[\log q(\beta)]=\begin{bmatrix}
\beta-\mu\\
\vech(\beta\beta'-\Sigma-\mu\mu') 
\end{bmatrix}.
\eeq
For the inverse gamma distribution $q(\tau^2)$ with density
\beqn
q(\tau^2)=\frac{a^b}{\Gamma(a)}(\tau^2)^{-1-a}\exp(-b/\tau^2),
\eeqn
the natural parameters are $(a,b)$.
The Fisher information matrix for the inverse gamma is
\beqn
I_F(a,b) = \begin{pmatrix}
\nabla_{aa}[\log \Gamma(a)]&-1/b\\
-1/b&a/b^2
\end{pmatrix}.
\eeqn
and the gradient 
\bean
\nabla_a[\log q_\l(\t)]&=& -\log(\tau^2)+\log(b)-\nabla_a[\log \Gamma(a)]\\
\nabla_b[\log q_\l(\t)]&=& -\frac{1}{\tau^2}+\frac{a}{b}.
\eean

\bibliographystyle{apalike}
\bibliography{references_v1}

\end{document}